\documentclass[11pt, oneside]{article}   	
\usepackage{geometry}                		
\usepackage{graphicx}				
\usepackage{amssymb}

\usepackage{amsthm}
\usepackage{subfigure}
\usepackage{prelimdraft}
\usepackage{tikz}
\newcommand{\complexity}[1]{\textsc{#1}}
\newcommand{\NP}{\complexity{NP}}
\renewcommand{\P}{\complexity{P}}
\newcommand{\NPC}{\complexity{NP-Complete}}
\newcommand{\NPH}{\complexity{NP-Hard}}

\newcommand{\avertex}{\textsc{$\alpha$-Balanced-Vertex-Separator}}
\newcommand{\asubgraphvertex}{\textsc{$\alpha$-Subgraph-Balanced-Vertex-Separator}}
\newcommand{\YES}{\textsc{Yes}}
\newcommand{\NO}{\textsc{No}}

\newcommand{\card}[1]{|#1|}

\newtheorem{lemma}{Lemma}

\title{Yet Another Graph Partitioning Problem is \NPH{}}
\author{Ryan H. Lewis}
\DRAFT						

\begin{document}
\maketitle
\abstract{
Recently a large number of graph separator problems have been proven to be \NPH{}. 
Amazingly we have found that \asubgraphvertex{}, an important variant, has been overlooked.
In this work ``Yet Another Graph Partitioning Problem is \NPH{}" we present the
surprising\footnote{Dear dense reader: We are being completely serious this was totally not to be expected.} result that \asubgraphvertex{} is \NPH{}. This is despite the fact that the constraints 
of our new problem are harder to satisfy than the original problem.
}

\section{Introduction}
A \emph{vertex separator} is a set of vertices whose removal from a graph results in a disconnected graph.
Usually, we call a vertex separator $\alpha$-\emph{vertex-balanced} if the largest connected component only contains
an $\alpha$ fraction of the total [remaining] vertices. Finding a small $\alpha$-vertex-balanced separator 
is \NPH{}~\cite{bj-vp-92,VertexSepOdd}. In this work we show that finding separators is still hard if we change this constraint.
We begin with a review of graphs and complexity theory.

A \emph{graph} $G = (V,E)$ is a set $V$ of \emph{vertices}, and a set $E \subseteq V \times V$  of \emph{edges}. A subgraph $H \subset G$ is a graph composed of a subset of the vertices and edges of $G$. The \emph{degree} of a vertex $v$ is the number of edges incident to $v$ and is denoted $\deg_G(v)$. A vertex \emph{separator} $I \subset V$ is a set of vertices whose removal from $G$ results in a disconnected graph. For each connected component in $G \setminus I$ define it's vertex set to be $V_i$ and $E_i$ it's edge set. Let $E_i^I$ be the set of edges between $V_i$ and $I$ in $G$ and $\beta_i$ to be the set of vertices in $V_i$ incident to a vertex in $I$ in $G$.   Recall that:
\[ \sum_{v \in V} \deg_G(v) = 2\card{E} \]

A \emph{decision problem} is a question in some formal system with a \YES{} or \NO{} answer. An algorithm is said to be of polynomial time if its running time is upper bounded by a polynomial expression in the size of the input for the algorithm. A decision problem is said to be of polynomial time if there is an algorithm for determining it's solution which is of polynomial time. The class of polynomial time decision problems is called \P{}. \NP{} is the subclass of decision problems with a \YES{} solution where the problem of verifying the correctness of this solution is in \P{}. Given two decision problems A and B a \emph{polynomial time many one reduction} or \emph{reduction} from A to B is a polynomial time computable mapping $f: A \rightarrow B$ which preserves \YES{}/\NO{} answers. 
\section{Background}
We now state the  \avertex{} decision problem.
\begin{description}
\item[\textsc{Problem:}]  \avertex{}
\item[\textsc{Instance:}] A degree at most three graph $G$
\item[\textsc{Goal:}] Find a vertex separator $(V_1,V_2,I)$ of $G$ such that: 
\[ \card{I} \textrm{ is minimized.}  \textrm{ subject to }  \max_i{(\card{V_i})}  \leq \alpha(\card{V}) \]
\end{description}
The problem \avertex{} is \NPH{} ~\cite{bj-vp-92}. Wagner and M\"uller proved a similar result where the balance 
criterion was stated as $\max_i{\card{V_i}} \leq \card{V \setminus I}$ but $G$ was explicitly 3-regular~\cite{VertexSepOdd}. Both authors use a local replacement argument to prove there result. The local replacement gadget used by Bui and Jones is shown in Figure~\ref{fig:orig-gadget}. We sketch this argument now.

Suppose a graph $G$ has $n$ vertices and $m \leq n^2$ edges.  Bui and Jones create a gadget as shown in Figure~\ref{fig:orig-gadget}, where there are $4n^2$ cycles each containing $4n^2$ vertices. The white notes are referred to as \emph{outlet nodes}. There are $n^2$ outlet nodes. They create a new graph $G^*$ in which each local replacement gadget replaces a vertex in the original graph, and an edge $(v_i,v_j)$ in $G$, becomes an edge in $G^*$ between outlet node $i$ on $v_j$ and outlet node $j$ on $v_i$. Outlet nodes which are incident to edges originally in $G$ are called \emph{marked}. Clearly this forces the modified graph $G^*$ to be degree at most three. The authors the argue that:
\tikzset{
    master/.style={
        execute at end picture={
            \coordinate (lower right) at (current bounding box.south east);
            \coordinate (upper left) at (current bounding box.north west);
        }
    },
    slave/.style={
        execute at end picture={
            \pgfresetboundingbox
            \path (upper left) rectangle (lower right);
        }
    }
}
\newcommand{\gscale}{1.9}
\begin{figure*}[h1]
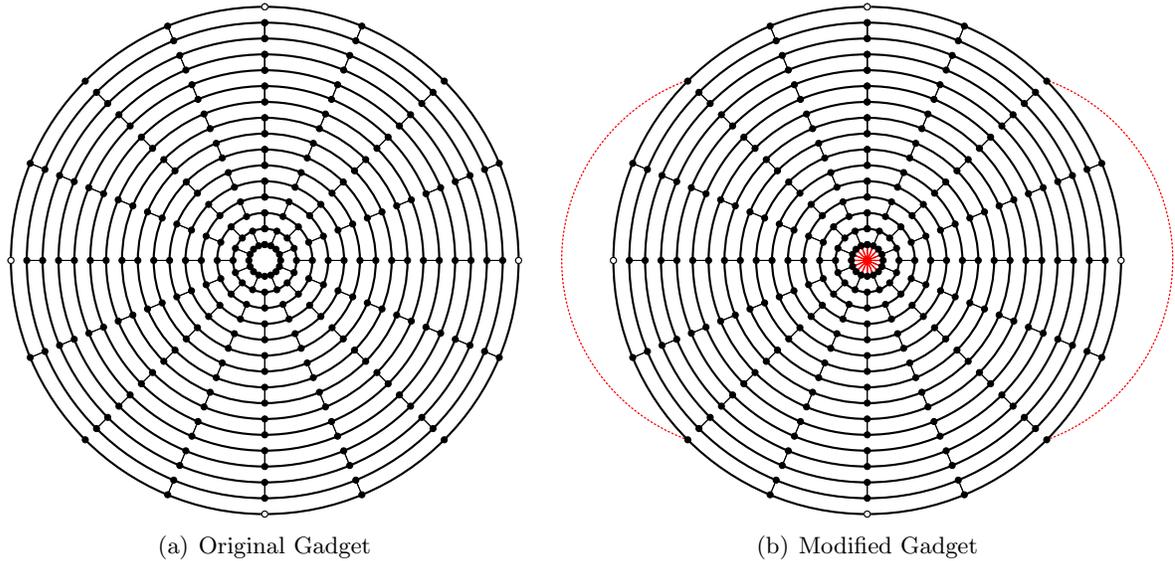

\centering
\subfigure[Original Gadget]{
	\begin{tikzpicture}[master,scale=\gscale,cap=round,rotate=90]
	\input{orig_gadget}
	\end{tikzpicture}
	\label{fig:orig-gadget}
}
\hfill
\subfigure[Modified Gadget]{
\begin{tikzpicture}[slave,scale=\gscale,cap=round,rotate=90]
\input{mod_gadget}
\end{tikzpicture}


\label{fig:mod-gadget}
}
\caption{(left) The local replacement gadget used by Bui \& Jones for a graph with two vertices. 
	 (right) We modify the original construction so that this operation produces a three-regular graph, instead of a graph with bounded degree three. To do this use gadget on the right in $G^*$ in place of the one on the left and then simply pair off remaining outlet nodes after installing the original edges from $G$. In the event that a vertex originally odd degree then add one more vertex between non outlet nodes.}
\label{fig:gadget}
\end{figure*}

\begin{enumerate}
\item Vertex separators in $G^*$ may be transformed into ones that are composed exclusively of outlet nodes.~(Page 157)
\item All minimal vertex separators in $G^*$ can be made to be independent sets~(Page 157)
\end{enumerate}
While not explicitly stated in the result it is easy to see that more guarantees on the separators of $G^*$ exist.  In particular, the authors assume the following property:
\begin{enumerate}
\item Any two vertices in a minimal vertex separator of $G^*$ are not adjacent to the same vertex.
\end{enumerate}
Finally, It is straightforward to modify the reduction to enforce that $G$ is exactly 3-regular.  The key is to slightly modify the local replacement gadget as show in Figure~\ref{fig:mod-gadget} and then to slightly modify it's use in the local replacement argument.  The modified gadget shows how only the unmarked outlet nodes within a gadget might be degree two. If there are even number of them, we may pair them off arbitrarily. If there are odd number, we add one more node to the outermost cycle in the appropriate place to pair this node with. All of the original arguments of Bui \& Jones hold so we do not repeat them here. We end this section by observing that this replacement argument has  a lot of freedom. For example by doubling the size of the gadget we can repeat each original edge in $G$ twice. This has the effect of doubling the degree of of every vertex. This not only handles the issue of odd degree vertices, but, also forces the size of a vertex separator to be even. In the next section we define our new vertex separator problem and get a reduction from the 3-regular variant of \avertex{} discussed in this section.

\section{Subgraph Balanced Vertex Separators}
We are interested in balanced vertex separators based on the size of their subgraphs. We state the our variant of the problem \asubgraphvertex{}. 
\noindent For any $\alpha \in (\frac{1}{2},1)$
\begin{description}
\item[\textsc{Problem:}]  \asubgraphvertex{}
\item[\textsc{Instance:}] A graph $G$
\item[\textsc{Goal:}] Find a vertex separator $(V_1,V_2,I)$ of $G$ such that: 
\[ \card{I} \textrm{ is minimized.} \] 
\textrm{ subject to } $ \max_i{(\card{V_i} + \card{E_i})} + \card{I} \leq \alpha(\card{V}+\card{E})  \textrm{ and } 
 I \textrm{ is an independent set.} $
\end{description}
 We now show that this problem is \NPH{} and it's decision problem variant \NPC{}.

From this point forward all graphs are subgraphs of 3-regular graphs. We assume that our graph instances have all the necessary properties discussed in the
previous section. We abuse notation
by dropping the notation $\card{.}$ to denote cardinality and denote the cardinality of a set by the
name of that set for brevity. We consider a graph instance $G^*$ of \avertex{} after the local replacement argument
and view it directly as an instance of \asubgraphvertex{}. To complete the reduction we need only show
that the two balance constraints are equivalent. That is, suppose we are given a \YES{} instance of \asubgraphvertex{} 
e.g we have a separator where 
\begin{equation}
 \max_i{(V_i + E_i + E_i^I)} +I \leq \alpha \frac{5}{2}V 
 \label{eqn:coverb}
 \end{equation}
We now show that this implies \[ \max_i(V_i) \leq \alpha V \]
and vice-versa. To begin we state a few lemmas.
\begin{lemma}
 \[ 3V_i = \sum_{v \in V_i} \deg_{G}(v) = 2E_i + E_i^I \]
 \label{lemma:one}
 \end{lemma}
 \begin{proof}
 This follows from the proof of the first theorem of graph theory where the edges in $E_i^I$ are only counted once and that our graph instance is 3-regular.
  \end{proof}
 \begin{lemma}
 \[ \beta_1 + \beta_2 = 3I \textrm{ and }  I \leq \beta_i \leq 2I  \]
 \label{lemma:two}
 \end{lemma}
 \begin{proof}
Since $G$ is 3-regular and $I$ is an independent set with the property mentioned its boundary is exactly $3I$ vertices. The bounds follow since each vertex
 must be adjacent to at least one vertex on either side.
 \end{proof}
 \noindent We now proceed show that the balance condition for \asubgraphvertex{} implies that of \avertex{}. 
Using Lemma~\ref{lemma:one} we simplify Equation~\ref{eqn:coverb}.
\[ \max_i{(4V_i - E_i)} + I =  \max_i{(V_i + E_i + E_i^I)} +I \leq \alpha \frac{5}{2}V \]
However now applying the first theorem of graph theory for the subgraph $G_i$ induced by $V_i$ we have
\[ 2E_i = \sum_{v \in V_i} \deg_{G_i}(v) = 3(V_i -\beta_i) + 2\beta_i = 3V_i - \beta_i \]
which we use to substitute $E_i = \frac{3}{2}V_i - \frac{1}{2}\beta_i$. We derive the following equation:
  \begin{equation}
 \max_i{(\frac{5}{2}V_i - \frac{1}{2}\beta_i)} +I  \leq \alpha \frac{5}{2}V 
 \label{eqn:two}
\end{equation}
Without loss of generality we assume that $\beta_2 \geq \beta_1$ and write

\[ \max{(\frac{5}{2}V_1 - \frac{1}{2}\Delta{\beta_i}, \frac{5}{2}V_2)} - \frac{1}{2}\beta_2 +I  =  \max_i{(\frac{5}{2}V_i - \frac{1}{2}\beta_i)} +I  \leq \alpha \frac{5}{2}V \]
where $\Delta{\beta_i} = \beta_2 - \beta_1$.  Using Equation~\ref{eqn:two} and Lemma~\ref{lemma:two}:
\[ \max(V_1, V_2) \leq \max{(V_1 + \Delta{\beta_i}/5, V_2)} \leq \max{(V_1 + \Delta{\beta_i}/5, V_2)} +\frac{2}{5}I- \frac{1}{5}\beta_2 \leq \alpha{V}  \]

\noindent We have shown that given a \YES{} instance of \asubgraphvertex{} is a \YES{} instance of \avertex{}. 

We now need to show the converse. Specifically if $\max(V_1, V_2) \leq \alpha{V}$ then  
\begin{equation}
\frac{5}{2}\max(V_1 + \frac{1}{5}\Delta{\beta_i}, V_2) - \frac{1}{2}\beta_2 \leq \alpha\frac{5}{2}V = \alpha(V+E)
\label{eqn:goal}
\end{equation}
Again we first assume that $V_1 + \frac{1}{5}\Delta{\beta_i} \leq V_2$. 
\[ \frac{5}{2}\max(V_1 + \frac{1}{5}\Delta{\beta_i}, V_2) - \frac{1}{2}\beta_2  = \frac{5}{2} \max(V_1,V_2) - \frac{1}{2}\beta_2 < \frac{5}{2}\max(V_1,V_2) \leq \frac{5}{2}\alpha{V} \]
Otherwise  if $V_1 + \frac{1}{5}\Delta{\beta_i}> V_2$
\[ \frac{5}{2}\max(V_1 + \frac{1}{5}\Delta{\beta_i}, V_2) - \frac{1}{2}\beta_2 =  \frac{5}{2}V_1 - \frac{1}{2}\beta_2 \leq \alpha\frac{5}{2}\max(V_1,V_2) \leq \frac{5}{2}\alpha{V} \]
This concludes the argument. 
\section*{Conclusion}
We have shown that \asubgraphvertex{} is \NPH{}. Yet another nail in the coffin for $\P{} = \NP{}$.
  \bibliographystyle{plain}
  \bibliography{yagppis}   

\end{document}